\documentclass[preprint,1p,11pt]{elsarticle}

\usepackage{amssymb}
\usepackage{amsthm}
\usepackage{amsmath}
\usepackage{amsfonts}
\usepackage{dsfont}
\usepackage{array}
\usepackage{epsfig}
\usepackage{tikz}
\usetikzlibrary{calc}
\usepackage{graphicx}
\usepackage{subfig}

\newtheorem{theorem}{Theorem}
\newtheorem{lemma}{Lemma}[section]
\newtheorem{definition}{Definition}
\newdefinition{remark}{Remark}
\newproof{pf}{Proof}

\usepackage[english]{babel}
\usepackage[T1]{fontenc}
\usepackage[utf8]{inputenc}

\newcommand{\bra}{\langle}
\newcommand{\ket}{\rangle}

\journal{Applied and Computational Harmonic Analysis}

\begin{document}

\begin{frontmatter}



\title{Sampling based on timing: Time encoding machines on shift-invariant subspaces}


\author[dg]{David Gontier}
\ead{david.gontier@ens.fr}
\address[dg]{Department of Mathematics, École Normale Supérieure (ENS)\\ 45 rue d'Ulm, 75005 Paris, France}

\author[mv]{Martin Vetterli}
\ead{martin.vetterli@epfl.ch}
\address[mv]{School of Computer and Communications Sciences, École Polytechnique Fédérale de Lausanne (EPFL)\\ CH-1015 Lausanne, Switzerland}

\begin{abstract}
Sampling information using timing is an approach that has received renewed attention in sampling theory. The question is how to map amplitude information into the timing domain. One such encoder, called time encoding machine, was introduced by Lazar and Tóth in \cite{Lazar04} for the special case of band-limited functions. In this paper, we extend their result to a general framework including shift-invariant subspaces. We prove that time encoding machines may be considered as non-uniform sampling devices, where time locations are unknown \textit{a priori}. Using this fact, we show that perfect representation and reconstruction of a signal with a time encoding machine is possible whenever this device satisfies some density property. We prove that this method is robust under timing quantization, and therefore can lead to the design of simple and energy efficient sampling devices.

\end{abstract}


\end{frontmatter}



\section{Introduction} 

Given a function $f(t)$ from a functional space $E$, sampling can be done in two different ways. In one approach, samples are taken at pre-defined sampling times $\{t_n\}_{n \in \mathbb{Z}}$, leading to a sequence $\{f_n\}_{n \in \mathbb{Z}} = \{f(t_n)\}_{n \in \mathbb{Z}}$. In this case, the question is when a set $\{t_n\}$ allows an exact representation of $f(t)$, or more fundamentally, a stable representation of $f(t)$. Examples of such sampling results are the Whittaker-Shannon sampling theorem \cite{Shannon49}, where $E$ is the set of band-limited functions $BL([-\Omega, \Omega])$ and $t_n = n (\pi / \Omega)$, sampling in shift-invariant subspaces \cite{Unser00}, or sampling of signals with finite rate of innovation (FRI) \cite{Vetterli02}.

Another, dual method, is sampling based on timing. Instead of recording the value of $f(t)$ at a preset time instant, one records the time at which the function takes on a preset value. Instead of the function itself, one may also consider the output of an operator $O[ \cdot ]$ applied to the function. For example, one may record the instants where $f(t)$ or $O[f](t)$ crosses a certain threshold.

Examples of such sampling methods include Logan's theorem for the representation of octave band functions from zero crossing \cite{Logan77}, various schemes generally known as \textit{delta-modulation}  \cite{Steel75}, as well as a method called \textit{time encoding machine} by Lazar and Tóth \cite{Lazar04} that mimics the \textit{integrate-and-fire} model of neurons. 

Why study such \textit{sampling with time} schemes ? On the one hand, the duality with respect to the more traditional \textit{sampling at preset time} makes it intriguing from a mathematical point of view. On the other hand, sampling by timing appears in nature, neurons being an example, and can lead to the implementation of very simple and low-cost sampling devices. A bucket that automatically empties itself once filled can be used as a time encoding machine to perform pluviometry. Finally, sampling by timing is potentially a more energy efficient way to acquire signals, since the basic elements (clocks, comparators, integrators,...) are lower power devices than high resolution analog-to-digital converters.

The purpose of the present paper is to analyze one class of \textit{sampling by timing} schemes, namely \textit{time encoding machines} of which the integrate-and-fire scheme of Lazar and Tóth \cite{Lazar04} is an example. The goal is to extend the validity of exact sampling by timing to broader classes of signals and operators, and to quantify robustness to timing quantization for these classes. By using the tools of frame theory and non-uniform sampling developed by Aldroubi, Feichtinger and Gröchenig \cite{Aldroubi98, Aldroubi01, Feich95, Groch92, Groch93, Schwab03}, it is possible to show that exact time encoding machines can be derived for a broad class of shift-invariant subspaces, and that these machines are robust to bounded timing noise.

The outline of the paper is as follows. Section 2 defines general time encoding machines (TEM) with two exemplary cases, crossing and integrate-and-fire TEMs. Section 3 reviews some properties of shift-invariant subspaces (SISS). In particular, the notion of reproducing kernel Hilbert spaces is reviewed. The main theorems linking density and invertibility of TEMs on SISS are presented in Section 4, as well as the fast reconstruction of the signal. Section 5 handles the special case of finite dimensional problems, to provide implementable algorithms. The question of stability under quantization noise is studied both theoretically and numerically in Section 6. Finally, possible research directions are indicated in Section 7.

\label{sec:Introduction}



\section{Time Encoding Machine}
\label{sec:TimeEncodingMachine}

We first extend the definition of Time Encoding Machine (or TEM) introduced by Lazar and Tóth \cite{Lazar04_1, Lazar06, Lazar04}:
\begin{definition}
	A \textit{Time Encoding Machine} is an operator $\mathcal{T}$ which maps a space $E$ of real valued functions to strictly increasing  sequences of reals:
	\[
		\begin{array}{lllll}
		\mathcal{T :} & E & \to & \mathbb{R}^{\mathbb{Z}} & \\
			& f(t) & \mapsto & \mathcal{T}f = \{t_n\} & \textnormal{with} \quad
				\left\{ \begin{array}{l}
						\quad \ldots < t_n < t_{n+1} < \ldots \\
						\lim\limits_{n \to \pm \infty} t_n = \pm \infty
					\end{array} \right.
		\end{array} . 
	\]
	The set $\{t_n\}$ denotes the \textit{sampling times}. We also say that $\mathcal{T}$ \textit{sampled} $f$ at time $t$ if $t \in \mathcal{T}f$. \\	
	We call a \textit{Time Decoding Machine} (or TDM) an operator $\mathcal{D}$ which maps an increasing sequence of reals into the space $E$. If $\mathcal{D} \circ \mathcal{T} = Id_E$, $\mathcal{T}$ is said to be \textit{invertible}, and $\mathcal{D}$ is an \textit{inverse} of $\mathcal{T}$. 
\end{definition}

The expression \textit{to sample} may be confusing, for no measure of amplitude is recorded. 
In practice, TEM are meant to encode a signal in real time, so that the fact that $\mathcal{T}$ samples $f$ at some time $t$ depends only on the set $\{ \left(t',f(t') \right), t' \le t\}$. We will only consider those types of TEMs in this paper. An important property of TEM is \textit{T-density}:
\begin{definition}
	A sequence $\{t_n\}_{n \in \mathbb{Z}}$ is $T$-dense if
	\[
		\forall n \in \mathbb{Z},  \quad t_{n+1} - t_n \le T.
	\]
	A TEM is $T$-dense if, for every input signals in $E$, the output is $T$-dense.
\end{definition}

We now present two common cases of TEMs, namely \textit{crossing TEM} and \textit{integrate-and-fire TEM}.


\subsection{Crossing TEM}
\label{ssec:CrossingTEM}

The study of TEMs is mainly motivated by their design and implementation. One easy example relies on a test function and a comparator only. The sampling times then correspond to the times when the signal crosses the test function (i.e. the difference of the two is null). An electronic implementation may be very energy efficient, for we do not need to measure amplitudes to encode the signal, but only the times when this amplitude is null, which requires a single comparator.

\begin{figure}[h]
\label{fig:TEM}
\hspace{-1.3 cm}
\subfloat[C-TEM without feedback]{
	\begin{tikzpicture}

		\node at (0,1.5) (f) [rectangle, draw, minimum height=25, minimum width=45] {$f(t)$};
		\node at (0,0) (psi) [rectangle, draw, minimum height=25, minimum width=45] {$\Psi(t)$};		
		\node at (3.5,0.8) (comparator) [rectangle, draw, minimum height=70, minimum width=80, text width=25mm] {\ \ Comparator $ f(t) = \Psi(t)$ ?};

		\node at (-1.5, 1.5) (-inf1) {};
		\node at (-1.5, 0) (-inf2) {};
		\node at (7, 0.8) (+inf) {};
		\node at (0, -0.8) {};
		\draw (-inf1) to (f);
		\draw (-inf2) to (psi);
		\draw[->] (f) to ( $(comparator.west) + (0,0.7)$);
		\draw[->] (psi) to ( $(comparator.west) - (0,0.8)$) ;
		\draw[->] (comparator) to (+inf);

		\foreach \x in {5.2, 5.7, 5.8, 6, 6.6} {
			\draw [->] (\x, 0.8) to (\x, 1.5);
		}
		\node at (5.8, 0.4) {$\{t_n\}$};
		
		\node at (2,-4) {\includegraphics[scale=0.5]{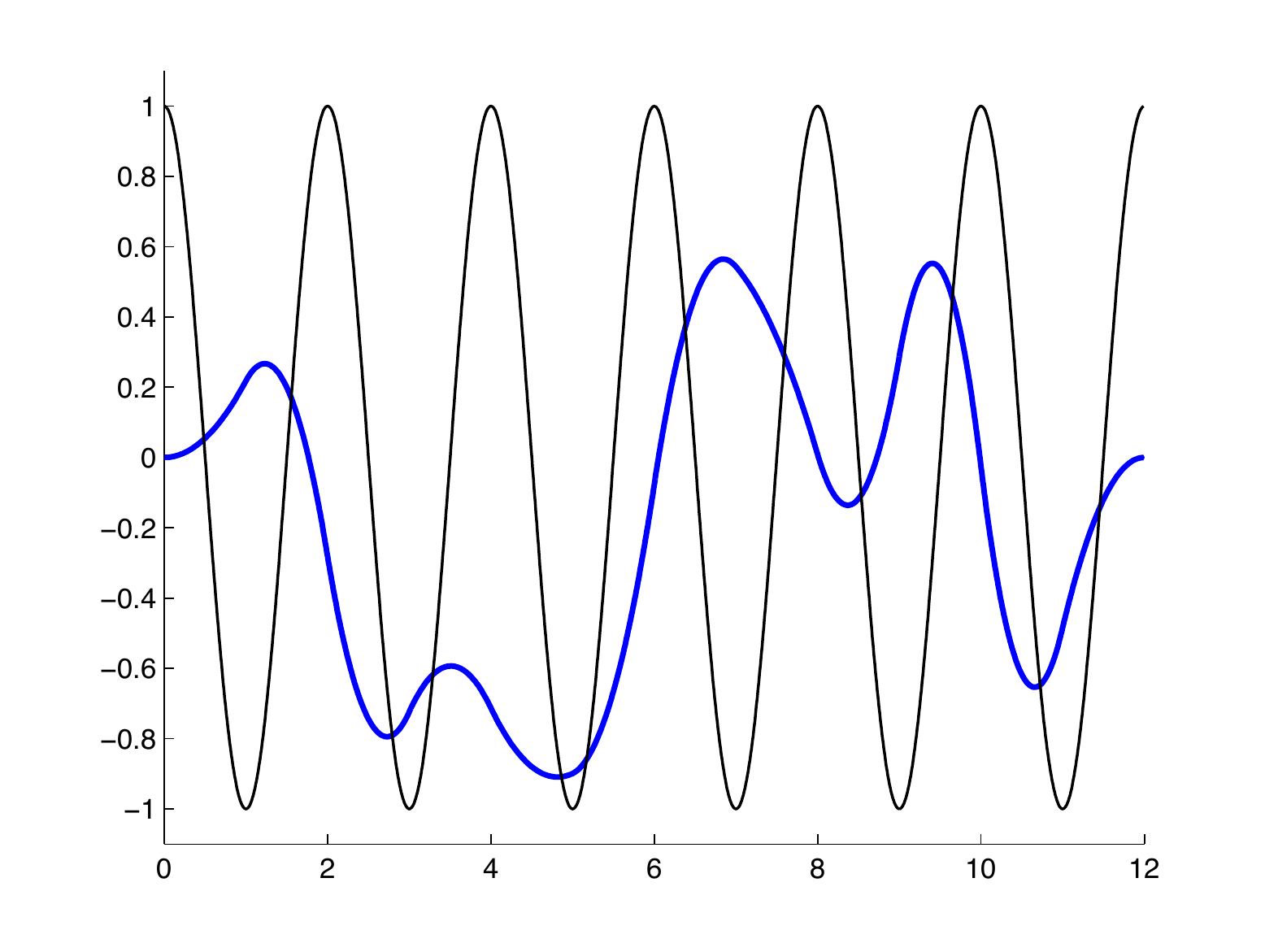}};
		\node at (2.2,-8) {\includegraphics[scale=0.33]{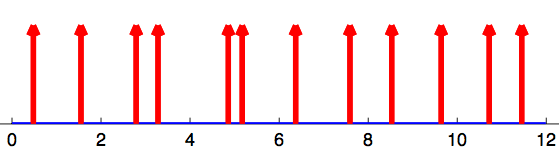}};

		\node at (5.6, -3.8) {\textcolor{blue}{$f(t)$}};
		\node at (5.6, -2) {$\Psi(t)$};
		\node at (2.7, -1.6) (tn) {$t_n$};
		\draw[->] (tn) to (2.9, -2.8);
		
		\node at (5.8, -8) {\textcolor{red}{$\{t_n\}$}};

	\end{tikzpicture}
}
\subfloat[C-TEM with feedback]{
	\begin{tikzpicture}

		\node at (0,1.5) (f) [rectangle, draw, minimum height=25, minimum width=45] {$f(t)$};
		\node at (0,0) (psi) [rectangle, draw, minimum height=25, minimum width=45] {$\Psi_n(t)$};		
		\node at (3.5,0.8) (comparator) [rectangle, draw, minimum height=70, minimum width=80, text width=25mm] {\ \ Comparator $ f(t) = \Psi_n(t)$ ?};

		\node at (-1.5, 1.5) (-inf1) {};
		\node at (-1.5, 0) (-inf2) {};
		\node at (7, 0.8) (+inf) {};
		\draw (-inf1) to (f);
		\draw (-inf2) to (psi);
		\draw[->] (f) to ( $(comparator.west) + (0,0.7)$);
		\draw[->] (psi) to ( $(comparator.west) - (0,0.8)$) ;
		\draw[->] (comparator) to (+inf);
		\draw (comparator.south) -- ( $(comparator.south) - (0,0.5)$) -- ( $(psi.south) - (0,0.5)$) edge[->] (psi.south);

		\foreach \x in {5.2, 5.7, 5.8, 6, 6.6} {
			\draw [->] (\x, 0.8) to (\x, 1.5);
		}
		\node at (5.8, 0.4) {$\{t_n\}$};
		
		\node at (2,-4) {\includegraphics[scale=0.5]{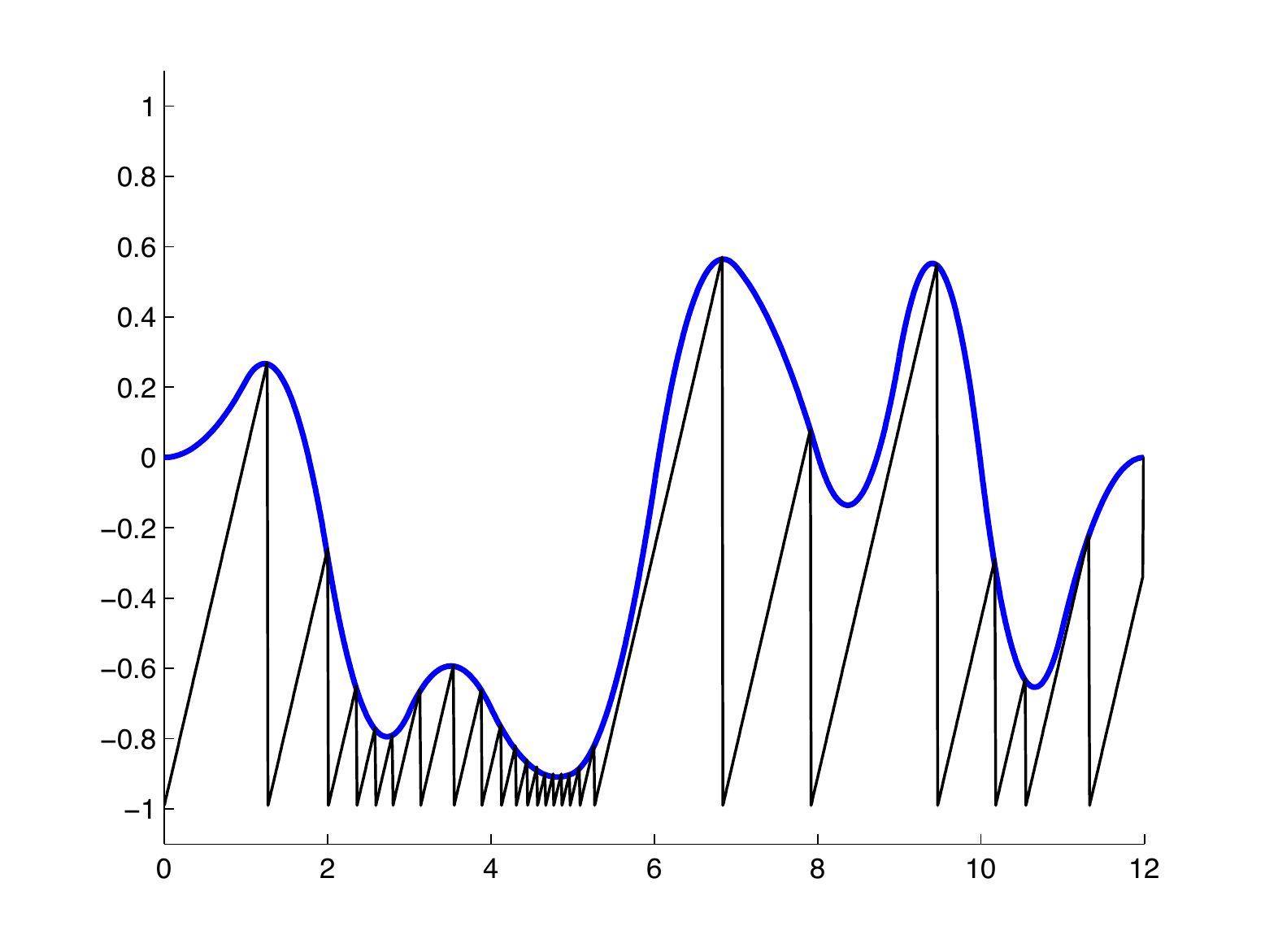}};
		\node at (2.2,-8) {\includegraphics[scale=0.33]{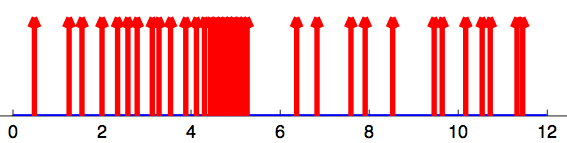}};
		
		\fill[white] (5.1, -4.8) rectangle (5.3, -3.91);
		
		\node at (5.6, -3.8) {\textcolor{blue}{$f(t)$}};
		\node at (5.6, -5.5) {$\Psi_n(t)$};
		\node at (4.7, -2.5) (tn) {$t_{n-1}$};
		\draw[->] (tn) to (4.8,-4.2);
		
		\node at (5.8, -8) {\textcolor{red}{$\{t_n\}$}};

	\end{tikzpicture}
}

\caption{Diagrams representing two types of Crossing Time Encoding Machine. The device (a) uses a cosine as a test function, and there is no feedback. The device (b) uses a uses linear test functions which are reset at each spike, it has a feedback.}

\end{figure}
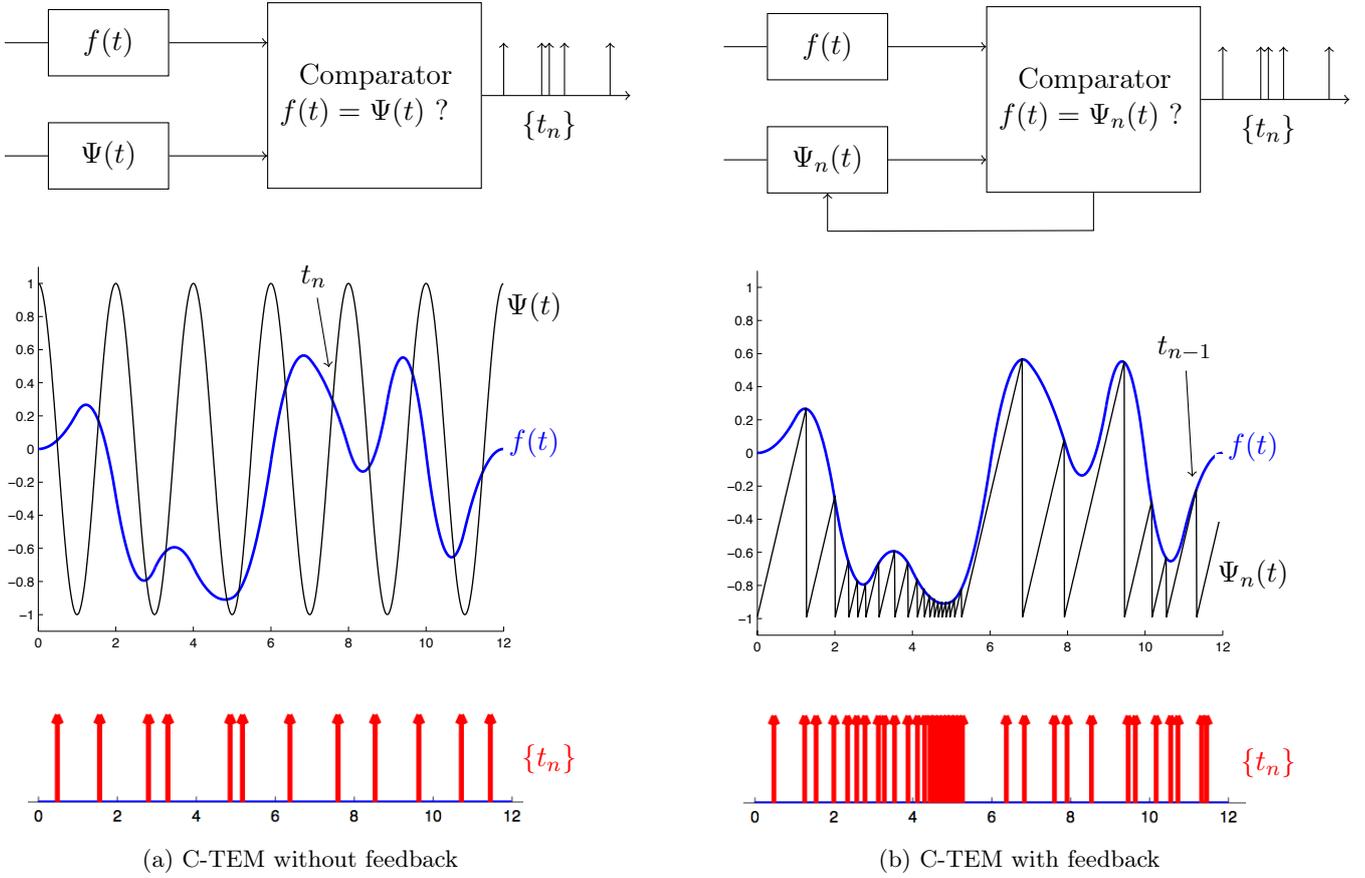

\begin{definition}
	A crossing TEM (or C-TEM) with continuous test functions $\{ \Phi_n \}$, denoted $\mathcal{CT}_{\{\Phi_n\}}$, outputs the sequence $\{t_n\} = \mathcal{CT}_{\{\Phi_n\}}  f$ such that:
\begin{itemize}
	\item $\Phi_n$ may be recovered from the set $\{ t_i, \ i \le n-1\}$
	\item $f(t_n) = \Phi_n(t_n)$
	\item $f(t) \neq \Phi_n(t) \quad \forall t \in ]t_{n-1}, t_n[$
\end{itemize}
\end{definition}
For instance, $\Phi_n$ may be a rapidly increasing function, which is replaced by $\Phi_{n+1}$ as soon as the TEM samples. In this case, the crossing TEM has a feedback, and $t_{n-1}$ may act as a time reference, so that a time shift between an encoder and a decoder will not affect the reconstruction. On the other hand, if $\Phi_n(t) = \Phi(t)$ is independent of $n$, then, the output is exactly:
\[
	\mathcal{CT}_\Phi f = \{ t \in \mathbb{R}, \ f(t) = \Phi(t) \} ,
\]
so that we don't need a feedback with those C-TEMs, but the knowledge of some fixed time reference is important. Those notions are illustrated in Figure 1. Let us introduce two implementations of such devices. \\

We suppose first that $E \subset \{ f \in C(\mathbb{R}), \ \| f\|_\infty \le \delta\}$ with $\delta < 1$. Then, we can use
\begin{equation} \label{eq:CosTestFunction}
	\Phi_{T} = cos(\frac{2\pi}{T}t)  .
\end{equation}
According to the mean value theorem, there exists a firing time in each $] k (T/2), (k+1)(T/2) [$, $k \in \mathbb{Z}$, so that $\mathcal{CT}_{\Phi_T}$ is $T$-dense. Furthermore, $\Phi_T$ is easily implementable with a RLC circuit, and potentially energy efficient. We could also have used
\begin{equation}
	\Psi_{n,T} = -1 + 2 \cdot \frac{(t - t_{n-1})}{T},
\end{equation}
so that $\Psi_{n,T}(t_{n-1}) = -1, \ \Psi_{n,T}(t_{n-1} + T) = 1$ and $\Psi_{n,T}$ is continuous and strictly increasing. Again, the mean value theorem states that $t_n \in [t_{n-1}, t_{n-1}+T]$, hence $\mathcal{CT}_{\{\Psi_{n,T}\}}$ is $T$-dense. In this case, $\Psi_{n,T}$ is linear, and can be approximated by the exponential charge of a capacitor which is reset at each spike, making it again potentially energy efficient.


\subsection{Integrate-and-Fire TEM}
\label{ssec:IFTEM}

The study of neurons leads to a different class of TEMs. It is shown that a neuron may be well-represented by an integrate-and-fire TEM \cite{Dayan01}\cite{Gerstner02}:
\begin{definition}
	An Integrate-and-Fire TEM (or IF-TEM) with test functions $\{ \Phi_n \}$, denoted $\ \mathcal{IT}_{\{\Phi_n\}}$, outputs the sequence $\{t_n\} = \mathcal{IT}_{\{\Phi_n\}} \ f$ such that:
\begin{itemize}
	\item $\Phi_n$ may be recovered from the set $\{ t_i, \ i \le n-1\}$
	\item $\int_{t_{n-1}}^{t_n} f(u) du = \Phi_n(t_n)$
	\item $\int_{t_{n-1}}^{t} f(u) du \neq \Phi_n(t) \quad \forall t \in ]t_{n-1}, t_n[$
\end{itemize}
\end{definition}

Therefore, an IF-TEM is nothing but a C-TEM on the integrated signal: $F(t) = \int_{-\infty}^t f(u)du$. In addition to providing a suitable model for neurons, it may be interesting to study those TEMs for integrated signals since they have nice properties, such as being continuous and differentiable. We will show in Section \ref{ssec:adjoint} that integrate-and-fire TEMs are adjoint to crossing TEMs. \\

\section{Non-Uniform Sampling and Shift-Invariant Subspaces}
\label{sec:NUS&SISS}


Throughout this study, we will work within the context of shift-invariant subspaces. The purpose of this section is to review the tools related to this framework. The notion of reproducing kernel Hilbert space will also be recalled, as sampling theory  requires the samples to be well-defined. All those notions but Lemma \ref{th:rkhs} may be found in \cite{Aldroubi98, Aldroubi01, Aldroubi02, Chen02, Dau92, Feich90, Feich92, Mallat08}.

\subsection{Shift-Invariant Subspaces}
\label{ssec:SISS}

\begin{definition}
	A Shift-Invariant SubSpace (or SISS) generated by the real function $\lambda(t)$ and of order $p$ is defined by
	\[
		V^p(\lambda) := \left\{ f(t) = \sum_{k \in \mathbb{Z}} c_k \lambda(t-k), \  (c_k)_{k \in \mathbb{Z}} \in l^p(\mathbb{R}) \right\} .
	\]
\end{definition}

We use the term shift-invariant, but only shifts over $\mathbb{Z}$ are considered in this definition. A special case is the space of band-limited functions $B := \left\{ f \in L^2(\mathbb{R}), Supp(\hat f) \subset [-\pi, \pi] \right\} = V^2(sinc_\pi)$, where $sinc_\pi(t) = sin(\pi t)/(\pi t)$, and $\hat f = \mathcal F f$ is the classic Fourier transform:
\[
	\hat f( \omega) = \int_{-\infty}^{\infty} f(t) e^{-i \omega t} dt .
\]
$B$ is invariant under every shift. This result comes directly from the Shannon sampling theorem \cite{Shannon49}. 
We will work mainly in the Hilbert space $L^2(\mathbb{R})$ ($p=2$). We recall (see \cite{Dau92} or \cite{Mallat08}) that
\begin{equation} \label{eq:ps}
	\forall f, g \in V^2(\lambda), \quad  \bra f \mid g \ket = \frac{1}{2 \pi} \int_{0}^{2\pi} \hat c (\omega) \overline{\hat d(\omega)} G_\lambda(\omega)^2 d \omega
\end{equation}
where $\hat c(\omega) =\sum c_k \exp (-i k \omega)$, $\hat d(\omega) =\sum d_k \exp (-i k \omega)$, and the coefficients $c_k, d_k$ are respectively the one of $f$ and $g$: $ f = \sum_{k \in \mathbb{Z}}c_k \lambda(\cdot - k)$ and $g = \sum_{k \in \mathbb{Z}}d_k \lambda(\cdot - k)$. We have defined
\[
	G_\lambda(\omega) := \left( \sum_{k \in \mathbb{Z}} \mid \hat \lambda(\omega + 2 k \pi) \mid^2 \right)^{1/2} .
\]
We deduce from this that
\begin{lemma} \label{le:equiv}
	The norms $\| \cdot \|_{L^2}$ and $\| f \|_{l^2} := \sum_{k \in \mathbb{Z}} | c_k |^2$ are equivalent if and only if
	\[ \exists 0 < A \le B< \infty \quad \text{such that} \quad 0 < A \le G_\lambda(\omega) \le B <\infty .\]
\end{lemma}
Those properties are also equivalent to the fact that $\{ \lambda( \cdot - k ) \}_{k \in \mathbb{Z}}$ is a frame of $V^2(\lambda)$ (see \cite[chap. 5]{Mallat08}). 
\subsection{Reproducing Kernel Hilbert space}

In order to calculate the crossings of a function from $V^2(\lambda)$ with a test function, we would like $f(t)$ to be well-defined for each $t \in \mathbb{R}$.
\begin{definition}
	A Hilbert space $\mathcal{H}$ is called a reproducing kernel Hilbert space (or RKHS) if the linear operators
	\[
			 \begin{array}{lcll}
			K_{x_0} : & \mathcal{H} & \to & \mathbb{R} \\
				& f & \mapsto & f(x_0)
			\end{array}
	\]
	are continuous, for all $x_0 \in \mathcal{H}$.
\end{definition}
In the case of $V^2(\lambda)$, this condition is easily obtained whenever $\lambda$ satisfies some weak conditions. For instance,  Aldroubi, Feichtinger and Gröchenig proved that it is enough to have $\lambda$ in the so-called amalgam Wiener space $W^1(C)$, i.e.
\[
	\left( \sum_k \textrm{ess} \sup_{x \in [0,1]} | \lambda(x+k) | \right)  < \infty .
\]
These authors have studied these spaces extensively in theory \cite{Feich90, Feich92} and in the non-uniform sampling framework \cite{Aldroubi98, Aldroubi01, Aldroubi02}. We will not review this theory, and use the different condition that $\lambda$ satisfies $ 0 < A \le G_\lambda \le B < \infty$ and is in the Sobolev space $H^1(\mathbb{R})$:
\[
	H^1(\mathbb{R}) := \left\{  f \in L^2, \| f\|_{H^1} < \infty     \right\} \quad \text{with} \quad \| f \|_{H^1}^2 = \| f \|^2_{L^2} + \| f' \|^2_{L^2} = \frac{1}{2 \pi} \int_\mathbb{R} (1 + \omega^2) | \hat f (\omega) |^2 d \omega .
\]
\begin{lemma} \label{th:rkhs}
	If $\lambda \in H^1(\mathbb{R})$ satisfies $ 0 < A \le G_\lambda \le B < \infty$, then $V^2(\lambda)$ is a RKHS, and we have $V^2(\lambda) \hookrightarrow L^\infty \cap C$.
\end{lemma}

\begin{proof}
	Let $f = \sum_{k \in \mathbb{Z}} c_k \lambda(\cdot - k)$ be an element of $V^2(\lambda)$. Let $x_0 \in [0,1]$ (the extension to $x_0 \in \mathbb{R}$ is similar), then we have:
	\begin{align*}
		f(x_0)^2 & =  \left( \sum_{k \in \mathbb{Z}} c_k \lambda(x_0 - k) \right)^2 \\
			& \le  \left( \sum_{k \in \mathbb{Z}} | c_k |^2 \right) \cdot \left( \sum_{k \in \mathbb{Z}} | \lambda (x_0 - k) |^2 \right) \quad\text{(Cauchy-Schwarz)}
	\end{align*}

	The first term is $\| f \|_{l^2}^2$, which is equivalent to $\| f \|^2_{L^2}$ according to Lemma (\ref{le:equiv}). The second one is bounded, as the 	following calculation shows.  We recall that $H^1(\mathbb{R}) \hookrightarrow C(\mathbb{R})$ (injection of Sobolev), so that $\lambda$ is continuous, and we can define $y_k = \textrm{arg}\min_{[-k, -k+1]} | \lambda | $. Then, with the Cauchy-Schwarz inequality, and the AM-GM inequality:
	\begin{align*}
		\sum_{k \in \mathbb{Z}} | \lambda (x_0 - k) |^2 & =  \sum_{k \in \mathbb{Z}} \left( \lambda(y_k)^2 + 2 \int_{y_k}^{x_0 - k} \lambda(t) \lambda'(t) dt \right) \\
		& \le \sum_{k \in \mathbb{Z}}  \lambda(y_k)^2  + 2 \sum_{k \in \mathbb{Z}} \left(\int_{y_k}^{x_0 - k} \lambda(t)^2 dt \right)^{1/2} \cdot \left(\int_{y_k}^{x_0 - k} \lambda'(t)^2 dt \right)^{1/2} \\
		& \le  \sum_{k \in \mathbb{Z}} \left( \int_{-k}^{-k+1} \lambda(y_k)^2 dt \right)  + \sum_{k \in \mathbb{Z}} \left(\int_{y_k}^{x_0 - k} \lambda(t)^2 dt \right) +  \sum_{k \in \mathbb{Z}} \left(\int_{y_k}^{x_0 - k} \lambda'(t)^2 dt \right) \\
		& \le \sum_{k \in \mathbb{Z}} \left( \int_{-k}^{-k+1} \lambda(t)^2 dt \right) +  \sum_{k \in \mathbb{Z}} \left(\int_{-k}^{- k+1} \lambda(t)^2 dt \right) +  \sum_{k \in \mathbb{Z}} \left(\int_{-k}^{- k+1} \lambda'(t)^2 dt \right) \\
		& = 2 \| \lambda \|_{L^2}^2 + \| \lambda' \|_{L^2}^2 \le 2 \| \lambda \|_{H^1}^2.
	\end{align*}
	We have therefore proved that $f(x_0) \le c \| \lambda \|_{H^1} \| f \|_{L^2}$, hence $V^2(\lambda)$ is a RKHS. \\
	To prove that $f$ is continuous, we write again~:
	\begin{align*}
		| f(x) - f(y) |^2 & \le \left( \sum_{k \in \mathbb{Z}} | c_k|^2 \right) \cdot \left( \sum_{k \in \mathbb{Z}} | \lambda(x - k) - \lambda (y-k) |^2 \right)\\
			& \le \| f \|_{l^2}^2 \cdot \sum_{k \in \mathbb{Z}}  \left( \int_{x-k}^{y-k} \lambda'(t) dt \right)^2 \\
			& \le \| f \|_{l^2}^2 \cdot |x-y| \cdot \sum_{k \in \mathbb{Z}} \int_{x-k}^{y-k} \lambda'(t)^2 dt
	\end{align*}
	which converges to $0$ whenever $y \to x$, so that $f$ is continuous.
\end{proof}
Note that because $f$ is continuous and bounded, if we normalize $f \in V^2(\lambda)$ so that $\| f \|_\infty = 1$, it is possible to construct a $T$-dense C-TEM for any $T>0$ (take for instance $\mathcal{CT}_{a \cos \omega t}$ with $a > 1$ and $\omega = \pi / T$).

%

In this case, point-wise evaluations are continuous functionals in this space. Therefore, the Riesz theorem states that for all $x \in \mathbb{R}$, there exists a function $K_x \in V^2(\lambda)$ such that:
\[
	\forall f \in V^2(\lambda), \ f(x) =  \bra f \mid K_x \ket.
\]
$K(y,x) = K_x(y)$ is called the reproducing kernel of the space. An expression of $K_x$ is given by

	\begin{equation} \label{eq:K}
		K_x(t) = \sum_{k \in \mathbb{Z}} \lambda(x-k) \tilde \lambda(t -k) = \sum_{k \in \mathbb{Z}} \tilde \lambda(x-k) \lambda(t -k) ,
	\end{equation}
where $\{\tilde \lambda(\cdot -k)\}_k$ is the bi-orthogonal frame of the frame $\{\lambda(\cdot -k)\}_k$:
\[
 	\tilde \lambda (t) := \mathcal{F}^{-1} \left( \frac{\hat \lambda (\omega)}{G^2_{\lambda} (\omega)} \right) \quad \text{so that} \quad \hat{\tilde \lambda}(\omega) = \frac{\hat \lambda (\omega)}{G^2_{\lambda} (\omega)}.
\]
This fact can be easily proved using (\ref{eq:ps}), so that we have the relations:
\[
	\bra \lambda(\cdot - k) \mid \tilde \lambda(\cdot - l) \ket = \delta_{kl} .
\]

This bi-othogonal frame is useful to write explicitly the projector $P_{V^2}$ on $V^2(\lambda)$ (which is closed whenever $0 < A \le G_\lambda \le B < \infty$):
\begin{equation} \label{eq:P}
		\forall f \in L^2, \ (P_{V^2} f) (t) =  \bra f \mid K_t \ket = \ \sum_{k \in \mathbb{Z}} \bra f \mid \tilde \lambda (\cdot - k) \ket \ \lambda(t-k),
\end{equation}
so that the $k^{th}$ coefficient of $(P_{V^2} f)$ in the $\{\lambda(\cdot -k)\}$ frame is $\bra f \mid \tilde \lambda(\cdot -k) \ket$.

\section{Time Encoding Machine on Shift-Invariant Subspaces}

After this review of non-uniform sampling on SISS, we are ready to analyze TEM in detail and give a sufficient conditions for these machines to be invertible.

\subsection{General analysis}

Time encoding machines provide an elegant way to encode amplitude information into the timing domain. Indeed, from the output $\{t_n\}_{n \in \mathbb{Z}}$ of a known Crossing-TEM, we can reconstruct the test functions $\Phi_{n}$, hence recover $f(t_n) = \Phi_{n}(t_n)$. Therefore, we can consider TEMs as a special method of non-uniform sampling, where the locations of the samples depend on the input function. In order to construct an approximation of the signal, we introduce the operator $V_{\{t_n\}}$ (denoted $V$) defined as follows. Let
\[
	s_n = \frac{t_{n-1} + t_n}{2}
\]
so that $\mathds{1}_{[s_n, s_{n+1}[}$ is the Voronoï region of $t_n$. Then $V f$ is defined by:
\begin{equation} \label{def:V}
	V f (t) = \sum_{n=\infty}^{\infty}  f(t_n) \cdot \mathds{1}_{[s_n, s_{n+1}[} (t).
\end{equation}
$V f$ is a piecewise constant function: $Vf (t) = f(t_n)$ for $t \in [s_n, s_{n+1}[$. The key fact is that, for a Crossing-TEM $\mathcal{CT}$, we can compute $Vf$ from the output $\{t_n\} = \mathcal{CT} f$. Hence, if we can invert $V$, we can recover the signal $f(t)$. This is the case whenever $\| Id - V \| < 1$ on a well-defined space. $\| f - V f \|_{L^2}$ is an integral of the difference between a function $f(t)$ and its approximation by a piece-wise linear function $Vf(t)$. According to the Riemann sums, we have the intuition that, provided that $t_n$ are dense enough, this difference should be small. The following calculations have been used by Gröchenig in \cite{Groch92}, then by Lazar and Tóth in \cite{Lazar04}, or Chen, Han and Jia in \cite{Chen04, Chen05}. We review them for completeness. We have:
\begin{eqnarray} \label{eq:1}
	\| f - V f \|^2_{L^2} & = & \int_{-\infty}^{\infty} | f(t) - Vf(t)|^2 dt = \sum_{n \in \mathbb{Z}} \int_{s_n}^{s_{n+1}} | f(t) - f(t_n)|^2 dt \nonumber \\
	 	& = & \sum_{n \in \mathbb{Z}} \int_{s_n}^{t_n} | f(t) - f(t_n)|^2 dt + \sum_{n \in \mathbb{Z}} \int_{t_n}^{s_{n+1}} | f(t) - f(t_n)|^2 dt .
\end{eqnarray}
We then use Wirtinger's inequality:
\begin{lemma} [Wirtinger's inequality]
	Let $f, f' \in L^2(a,b)$ with $f(a) = 0$ or $f(b)=0$, then:
	\[
		\int_a^b f(u)^2 du \le \frac{4}{\pi^2} (a-b)^2 \int_a^b f'(u)^2 du .
	\]
\end{lemma}

\begin{proof}
	We calculate, for $f \in L^2(0, \pi/2)$ with $f(0) = 0$,
	\[
		\int_a^b (f'(u)^2 - f(u)^2) du \ = \int_0^{\pi/2} [ f'(u) - f(u)\cot(u)]^2du - \underbrace{\int_0^{\pi/2} (f^2(u) \cot(u))' du}_0 \ge 0 ,
	\]
	and we obtain the general inequality with a change of variables.
\end{proof}
We suppose now that $f' \in L^2$, so that using Wirtinger's inequality in (\ref{eq:1}) leads to

\begin{equation} \label{eq:2}
	\| f -Vf\|^2_{L^2} \le \sum_{n \in \mathbb{Z}} \frac{4}{\pi^2} (t_n - s_n)^2\int_{s_n}^{t_n} |f'(t)|^2 dt + \frac{4}{\pi^2} (s_{n+1} - t_n)^2 \int_{t_n}^{s_{n+1}} | f'(t)|^2 dt . 
\end{equation} 
We can see that this error depends on the differences $(t_n - s_n)$, or $(t_{n+1} - t_{n})$. If we suppose that $(t_{n+1} - t_n) < T$, so that $(t_n - s_n) < T/2$ and $(s_{n+1} - t_n)<T/2$, equation (\ref{eq:2}) becomes: 

\begin{equation}\label{eq:3}
	\| f - Vf \|_{L^2}^2 \ \le \ \frac{T^2}{\pi^2} \sum_{n \in \mathbb{Z}} \int_{s_n}^{t_n} | f'(t)|^2dt + \int_{t_n}^{s_{n+1}} | f'(t)|^2dt \ \le \ \frac{T^2}{\pi^2} \| f' \|^2_{L^2} .
\end{equation}
Finally, if $f \in V^2(\lambda)$, and $\lambda' \in L^2$, we can use a final inequality:

\begin{lemma} [Chen, Han, Jia \cite{Chen05}] \label{th:chen}
	Let $\lambda \in H^1(\mathbb{R})$, then,
	\[
		\forall f \in V^2(\lambda), \ \| f' \|_{L^2} \le \left( \ \sup_{\omega \in [0, 2\pi[} \frac{G_{\lambda'}(\omega)}{G_\lambda(\omega)} \right) \ \| f\|_{L^2} .
	\]
\end{lemma}

\begin{proof}
	Let $f(t) = \sum_k c_k \lambda(t-k)$, so that $f'(t) = \sum_k c_k \lambda'(t-k)$. We use the formula of the scalar product in $V^2(\lambda)$ described in (\ref{eq:ps}): we can write $\hat c(\omega) = \sum_k c_k \cdot \exp(-ik\omega)$, so that:
	\[
		\| f\|_{L^2}^2 = \frac{1}{2 \pi} \int_0^{2 \pi} | \hat c(\omega)|^2 \cdot G_\lambda(\omega)^2 d \omega .
	\]
	In a similar way, we have:
	\[
		\| f'\|_{L^2}^2 = \frac{1}{2 \pi} \int_0^{2 \pi} | \hat c(\omega)|^2 \cdot G_{\lambda'}(\omega)^2 d \omega \le \frac{1}{2 \pi} \int_0^{2 \pi} | \hat c(\omega)|^2 \cdot G_{\lambda}(\omega)^2 \cdot \left( \frac{G_{\lambda'}(\omega)}{G_{\lambda}(\omega)} \right)^2 d \omega
		\le \left( \ \sup_{\omega \in [0, 2\pi[} \frac{G_{\lambda'}(\omega)}{G_\lambda(\omega)} \right)^2 \| f \|_{L^2}^2 .
	\]

\end{proof}
Plugging this into (\ref{eq:3}) leads to:
\begin{equation} \label{eq:4}
	\| f - Vf \|_{L^2}^2 \le \frac{T^2}{\pi^2} \| f '\|^2 \le \frac{T^2}{\pi^2} \cdot \left( \sup_{\omega \in [0, 2\pi[} \frac{G_{\lambda'}(\omega)}{G_\lambda(\omega)} \right)^2 \cdot \| f \|_{L^2}^2.
\end{equation}
We are now able to state our main theorem.


\subsection{Sufficient conditions for invertibility of TEMs}

We define
\[
	E = \left\{ \lambda \in H^1(\mathbb{R}), \quad 0 < A \le G_{\lambda}(\omega) \le B < \infty \quad \text{and} \quad G_{\lambda'} ( \omega) < \infty \right\}
\]
as a shorthand to have $\lambda$ satisfy the condition of Lemma \ref{th:rkhs} so that $V^2(\lambda)$ is a RKHS, with the extra condition:
\[
	\sup_\omega \frac{G_{\lambda'}(\omega)}{G_\lambda(\omega)} < \infty .
\]
In the following, we will denote $P := P_{V^2}$ the projection on $V^2(\lambda)$ as defined in (\ref{eq:P}).

\begin{theorem} \label{th:main}
	Let $\lambda \in E$ and
	\[
		\tau = \pi \cdot \inf_{\omega} \frac{G_{\lambda}(\omega)}{G_{\lambda'}(\omega)} > 0,
	\]
	then, for all $T < \tau$, every Crossing-TEM $\mathcal{CT}$which is $T$-dense is invertible. Moreover, if $\{t_n\} = \mathcal{CT}f$, an iterative reconstruction is given by the sequence:
	\begin{equation} \label{eq:reconstruction}
		\begin{array}{lll}
			f_1 & = & PV f_0 \\
			f_{k+1} & = & f_1 + (Id - PV) f_k
		\end{array} ,
	\end{equation}
	where $P$ is the orthogonal projector on $V^2(\lambda)$ and $V = V_{\{t_n\}}$ is defined in (\ref{def:V}). We have:
	\[
		\| f - f_k \|_{L^2} \le \Big( \frac{T}{\tau} \Big)^k \| f \|_{L^2} \xrightarrow[k \to 0]{} 0 .
	\]
\end{theorem}

\begin{proof}
	With the above notations, using (\ref{eq:4}), we have:
	\[
		\| f - PV f \|_{L^2} = \| P(f - PVf) \|_{L^2} = \| f - Vf \|_{L^2} \le \frac{T}{\tau} \| f \|_{L^2}, \quad \forall f \in V^2(\lambda).
	\]
	The operator $(PV): V^2(\lambda) \to V^2(\lambda)$ verifies $\| Id - PV \| \le T/\tau < 1$ whenever $T < \tau$, hence is invertible. We then show that $f - f_k = ( Id - PV)^k f$ by induction. We have indeed:
	\[
		 f - f_{k+1}= f - f_1 - (Id - PV)f_k = (Id - PV) f - (Id - PV) f_k = (Id - PV)(f - f_k)= (Id - PV)^{k+1} f ,
	\]
	so that $\|f - f_k\|_{L^2} \le (T/\tau) \| f \|_{L^2}$. Finally, because $f_1$, hence all the $f_k$ can be calculated from the output $\{t_n\}$ of the Crossing-TEM, this TEM is invertible.
\end{proof}

In the case $\lambda = sinc_{\pi}$ for instance, we have $\hat \lambda = \mathds{1}_{[-\pi, \pi[}$, so that $G_{\lambda} = 1$.  Then, we have $\hat {\lambda'}(\omega) = \omega \mathds{1}_{[-\pi, \pi[}$, hence $(\inf G_\lambda / G_{\lambda'}) = \pi$. Our theorem states in this case that if $T < 1$, then a $T$-dense TEM on band-limited functions is invertible. We find the Nyquist bound, which is (almost) the optimal bound we could hope for, according to the Beurling-Landau theorem. The reader may find a complete study of the band-limited case in \cite{Beutler66, Landau67} or in the very good book \cite{Young80}. In the general case, this bound is not optimal. Chen, Han and Jia discussed this bound in \cite{Chen04}. Chen, Itoh and Shiki studied the case of wavelet subspaces in \cite{Chen98}. Aldroubi and Feichtinger have calculated an optimal bound for cardinal spline spaces in \cite{Aldroubi99}, which is better than the one given with this theorem. Gröchenig has also studied the case where $\lambda \in H^2(\mathbb{R})$ in \cite{Groch92}, and derived another bound using a different Wirtinger inequality. He also noticed that one can speed up the convergence using piecewise linear functions instead of piecewise constant functions for the operator $V$ (we get a convergence in $O(\beta^n)$ instead of $O(\alpha^n)$, with $\beta < \alpha$).

The operator $(PV)$ consists in calculating a piecewise constant function approximating $f \in V^2(\lambda)$ on $\{t_n\}$, and projecting the result back on $V^2(\lambda)$. This theorem states that the operator $(PV)$ is invertible whenever $\{t_n\}$ is dense enough. It links therefore density and invertibility of TEMs, and provides an iterative method to reconstruct the signal with an exponential speed of convergence. We could also have calculated the inverse directly: let $f \in V^2(\lambda)$ be represented by the infinite column vector $\textbf{c} = (\ldots, c_k, c_{k+1}, \ldots)^T$ of its coefficients: $f(t) = \sum_k c_k \lambda(t-k)$, then

\begin{lemma} \label{th:reconstruction}
	Let $M$ and $A$ be infinite matrices over $\mathbb{Z} \times \mathbb{Z}$ with coefficients $M_{jk} = \lambda(t_j - k)$ and $A_{ki} = \int_{s_i}^{s_{i+1}} \tilde \lambda(u-k)du$. If $f(t) = \sum_k c_k \lambda(t-k)$ and $g(t) = PV f = \sum_k d_k \lambda(t-k)$, then $\textbf{d} = AM \textbf{c}$.
\end{lemma}

\begin{proof}
	We first have:
	\[
		f_1(t) := (Vf)(t) = \sum_{j \in \mathbb{Z}}  f(t_j) \cdot \mathds{1}_{V_j} (t)
	\]
	and according to (\ref{eq:P}),
	\[
		g(t) = (PV f) (t) = (P f_1) (t) = \sum_{k \in \mathbb{Z}} \sum_{j \in \mathbb{Z}} f(t_j) \cdot \bra \mathds{1}_{V_j} \mid \tilde \lambda(\cdot -k) \ket \lambda(t-k) .
	\]
	We check that $\bra \mathds{1}_{V_j} \mid \tilde \lambda(\cdot -k) \ket = \int_{s_j}^{s_{j+1}} \tilde \lambda(u-k)du = A_{kj}$ and that $f(t_j) = \sum_k c_k \lambda(t_j -k) = (M \textbf{c})_j$. Altogether,
	\[
		d_k = \sum_j f(t_j) \cdot \bra \mathds{1}_{V_j} \mid \tilde \lambda(\cdot -k) \ket = \sum_j A_{kj} (M\textbf{c})_j 
	\]
	which can also be written $\textbf{d} = AM\textbf{c}$.
\end{proof}

The question of whether we can recover a signal $f$ from its samples $f(t_n) = \bra f \mid K_n \ket$ has a classic answer in frame theory: we can do so whenever there exists $0 < A \le B < \infty$ such that for all $f \in V^2(\lambda)$, $A \| f \|^2_{L^2} \le \sum f(t_n)^2 \le B \| f \|_{L^2}$, i.e. $\{K_{t_n}\}$ is a frame. In particular, if there exists $0 < \delta \le T$ such that $\delta \le t_{n+1} - t_n \le T$, then this is the case (see \cite{Aldroubi01}), and we can recover the signal in a stable way with:
\[
	f(t) = \sum_{k \in \mathbb{Z}} f(t_n) \tilde K_{t_n}(t)
\]
where $\{\tilde K_{t_n}\}$ is the bi-orthogonal frame of $\{K_{t_n}\}$. In our case, we are not using the fact that $t_{n+1} - t_n \ge \delta$, but the weaker form~: $\lim_{n \to \pm \infty} t_n = \pm \infty$. This may lead to an unstable reconstruction, but we can always deduce a stable one by erasing some extra data, according to this result. Instead of deleting data, we have chosen the more elegant way of weighting the samples (see \cite{Groch92}). Indeed, we can show that:
\[
	A \| f \|^2_{L^2} \le \sum_{n \in \mathbb{Z}} w_n f(t_n)^2 = \| V f \|^2_{L^2}  \le B \| f \|^2_{L^2} \quad \text{with} \quad w_n = s_{n+1} - s_n .
\]

\subsection{Integrate-and-Fire TEM as the adjoint of crossing TEM}
\label{ssec:adjoint}

We show in this section that IF-TEMs may be seen as the quasi-adjoint case of C-TEMs. Indeed, with an IF-TEM $\mathcal{IT}$, we may now recover $\Phi_n(t_n) = \int_{t_n}^{t_{n+1}} f(u) du$ from $\{t_n\}_{n \in \mathbb{Z}} = \mathcal{IT} f$. We introduce the operator $Z = Z_{\{t_n\}}$ defined by
\[
	(Z f)(t) := \sum_{n \in \mathbb{Z}} \int_{t_n}^{t_{n+1}} f(u) du \cdot K_{s_{n+1}}(t).
\]
We do not comment on the spaces on which $Z$ acts, for we are mainly interested into the operator $(PZ)$ which will have a meaning. Note for instance that $P : L^2 \to V^2(\lambda)$ can be defined on a larger space than $L^2$, namely $V^2(\lambda)'$, the dual space of $V^2(\lambda)$.  Let also introduce the operator
\[
		(V'f)(t) := \sum_{n \in \mathbb{Z}} f(s_{n+1}) \cdot \mathds{1}_{[t_n, t_{n+1}[}(t).
\]
\begin{lemma}
	If $\{t_n\}$ is T-dense, with $\lim_{n \to \pm \infty} t_n = \pm \infty$, then 
	\begin{itemize}
		\item $V'$ is well-defined and linear continuous as an operator from $H^1$ to $L^2$.
		\item If $\lambda \in E$, $(P V')$ is an operator from $V^2(\lambda)$ to $V^2(\lambda)$, and we have:
		\[
			\forall f, g \in V^2(\lambda), \quad \bra f \mid (PV')g \ket = \bra (PZ)f \mid g \ket
		\]
	\end{itemize}
\end{lemma}

\begin{proof}
	Let $f \in H^1(\mathbb{R})$, then
	\[
		\| V'f \|_{L^2}^2 = \sum_{n \in \mathbb{Z}} (t_{n+1} - t_n) f(s_{n+1})^2 .
	\]
	As in the proof of Lemma \ref{th:rkhs}, we introduce $y_{n+1} = \min_{[t_n, t_{n+1}]} | f |$, so that
	\begin{align*}
		\| V'f \|_{L^2}^2 & = \sum_{n \in \mathbb{Z}} (t_{n+1} - t_n) \cdot \left( f(y_{n+1})^2 + \int_{y_{n+1}}^{s_{n+1}} 2 f(u) f'(u) du \right) \\
			& \le  \sum_{n \in \mathbb{Z}} (t_{n+1} - t_n) f(y_{n+1})^2 + T \sum_{n \in \mathbb{Z}} \left( \int_{t_n}^{t_{n+1}} f(u)^2 du + \int_{t_n}^{t_{n+1}} f'(u)^2 du \right) \\
			& \le (T+1) \| f \|_{L^2}^2 + T \| f' \|_{L^2}^2 \le (T+1) \| f \|_{H^1}^2 .
	\end{align*}
	which proves the first point. We notice then that if $\lambda \in E$, then $V^2(\lambda) \in H^1$, so that $(PV')$ is a continuous operator from $V^2(\lambda)$ to itself. Finally, for $f,g \in V^2(\lambda)$,
	\begin{align*}
	\bra f \mid (PV') g \ket & = \bra P f \mid (P V') g \ket = \bra f \mid V' g \ket = \int f(u) \sum_n \mathds{1}_{[t_n,t_{n+1}[} (u) \ \bra g \mid  K_{s_{n+1}} \ket du \\
	 	& = \sum_{n \in \mathbb{Z}} \int_{t_n}^{t_{n+1}} f(u) \int_{-\infty}^{\infty} g(t) K_{s_{n+1}}(t) dt \\
		& = \int_{-\infty}^{\infty} g(t) \Bigl( \sum_n \int_{t_n}^{t_{n+1}} f(u) du \cdot K_{s_{n+1}}(t) \Bigr) dt = \bra Zf \mid g \ket = \bra (PZ) f \mid g \ket
	\end{align*}
	
\end{proof}

The operator $V'$ is very close to the operator $V$ introduced previously: only the role of $t_n$ and $s_n$ has been permuted. But because the only important property of $V$ used to prove the invertibility of $(PV)$ is that $s_{n+1} - s_n < T$ whenever $t_{n+1} - t_n < T$, and because of the tautology $t_{n+1} - t_n < T$ whenever $t_{n+1} - t_n < T$, we can conclude as in Theorem \ref{th:main}:

\begin{theorem} \label{th:mainbis}
		Let $\lambda \in E$ and
	\[
		\tau = \pi \cdot \inf_\omega \frac{G_{\lambda}(\omega)}{G_{\lambda'}(\omega)} ,
	\]
	then, for all $T < \tau$, every IF-TEM which is $T$-dense is invertible. Moreover, the reconstruction is similar to the one of Theorem \ref{th:main}, using the operator $Z$ instead of $V$.
\end{theorem}

\begin{proof}
	Similar to Theorem \ref{th:main}. Notice that, restricted on $V^2(\lambda)$, $(Id - PZ)^* = (Id - PV')$, so that $\| Id - PZ \|_{V^2} = \| Id - PV' \|_{V^2} < 1$, and the same reconstruction method works.
\end{proof}

This duality has been noticed by Lazar and Tóth in \cite{Lazar04} in the case of $V^2(sinc_\pi)$, and we have extended it for all SISS using only properties from RKHS. \\

In the rest of the paper, we will only consider Crossing-TEM, due to this duality.


\subsection{Fast algorithm for reconstruction}
\label{ssec:algorithm}

We now present an algorithm which does not need the knowledge of $\tilde \lambda$ to work. Indeed, according to Lemma \ref{th:reconstruction}, we have to calculate $A_{ki} = \int_{s_j}^{s_{j+1}} \tilde \lambda(u) du$. But $\tilde \lambda$ may have a large support, so that those elements are computationally expensive to calculate. An alternative has been given by Gröchenig in \cite{Groch93}. We recall that $M$ is the matrix with elements $M_{ik} = \lambda(t_i -k)$. We suppose that the output $\{t_n\}$ is $T$-dense, and we denote $\tau = \pi \cdot \inf (G_\lambda/ G_{\lambda'})$ as in Theorem \ref{th:main}.

\begin{lemma} [Gröchenig \cite{Groch93}] \label{th:reconstruction2}
	Let $D$ be the diagonal matrix with elements $d_{ii} = w_i = s_{i+1} - s_i$, let $M^*$ be the adjoint of $M$, then $U := (M^* D M)$ is an operator from $V^2(\lambda)$ to itself, and for all $\mu \in \mathbb{R}$,
	\[
		\| Id - \mu U \|_{V^2}  \le \max \left\{1 - \mu(1-\gamma)^2, \mu( 1+\gamma)^2 -1 \right\} \quad \text{with} \quad \gamma = \frac{T}{\tau} .
	\]
	In particular, if $T < \tau$, with the optimal choice $\mu = (1 + \gamma^2)^{-1}$, we find 
	\[
		\| Id - \mu U\|_{V^2} \le \beta := \frac{2 \gamma}{1 + \gamma^2} < 1 .
	\]
\end{lemma}
Note that the speed of convergence is slower (we have a convergence in $O(\beta^n)$ instead of $O(\alpha^n)$, and $\beta > \gamma$), but each iterative step is faster to calculate: the matrix $A$ is no longer required; we only need to compute the diagonal matrix $D$ and $M$, which is inexpensive.

As we did before, we conclude that $U$ is invertible, and it gives a reconstruction of the signal with an iterative method. We have actually proved that $M$ is left-invertible as soon as $T < \tau$, and a left inverse of $M$ may be $M^\ddagger  = U^{-1}M^*D = ( \mu M^* D M)^{-1} \mu M^* D$. Of course, we could have used the pseudo-inverse $M^\dagger = (M^* M)^{-1} M^*$, but the inverse of $(M^*M)$ is more difficult to calculate, and less stable with respect to numerical error. We also notice that we introduced the matrix $D$ to compensate an accumulation of $\{t_n\}$.


\section{Periodic case}

As we have seen, TEMs allow to encode a large class of signals in real time, without loss of information. However, the decoding cannot be executed in real time. We also have considered infinite length signals in our model, which is not practical from an algorithmic point of view. In this section, we focus on the finite dimensional case. For instance, the signal may be in a shift-invariant subspace and $K$-periodic, so that it can be described with a finite number of coefficients. The signal could also be with finite support, but this case can be handled as a periodic case, which is why we will focus on periodic signals. Studying the finite dimensional case will also allow us to estimate the effects of noise. 

\subsection{Shift-invariant subspace and TEMs on $L^2_K$}

We will work in the Hilbert space
\[
	L^2_K(\mathbb{R}) := \{f, \forall x \in \mathbb{R}, f(x + K) = f(x), \int_0^K f^2(x) dx < \infty\}
\] 
together with the scalar product
\[
	\bra f \mid g \ket \ = \frac{1}{K} \int_0^K f(x) \bar g(x) dx  .
\]
The Fourier series theorem states that $f \in L^2_K(\mathbb{R})$ may be written as
\[
	f(t) = \sum_{n \in \mathbb{Z}} \hat f_n \cdot \exp \left( \frac{2i \pi n}{K}t \right),\quad \text{with} \quad \hat f_n = \bar{ \hat f}_{-n} = \bra f(t) \mid \exp \left( \frac{2 i \pi n}{K}t \right)  \ket .
\]

We introduce, for $\lambda \in L^2_K(\mathbb{R})$, the periodic shift-invariant subspace:
\[
	V^2_K(\lambda) := \{ f(t) = \sum_{k=1}^K c_k \lambda(t-k), \ c_k \in \mathbb{R} \} .
\]
The space is now of finite dimension, hence it is always closed in $L^2_K(\mathbb{R})$. The two norms $\| f \|_{l^2} := \left( \sum_{k=1}^K | c_k |^2 \right)^{1/2}$, where $f = \sum_{k=1}^K c_k \lambda( \cdot -k)$, and $\| f \|_{L^2_K}$ are equivalent whenever the map
\[
	\begin{array}{lll}
		\mathbb{R}^K & \to & V^2_K(\lambda) \\
		(c_k) & \mapsto & \sum_{k=1}^K c_k \lambda(t-k)
	\end{array}
\]
is injective.

%
%

Because we would like to recover the signal working only on $[-K/2, K/2]$, we suppose that TEMs samples only within this interval. The output is therefore $\{t_n\}_{n=1..J} = \mathcal{T} f$, where $-K/2 \le t_1 < t_2 < \cdots < t_J < K/2$. We introduce $t_{J+1} := t_1 + K$, and we modify slightly the definition of $T$-density, to handle periodicity:
\begin{definition}
	An increasing sequence $\{t_1, t_2, \cdots, t_J\}$ is $T$-dense if 
	\[
		t_{n+1} - t_n \le T, \quad \forall n = 1 ..  J .
	\]
\end{definition}
With this definition, all the previous results hold: a TEM on $V^2_K(\lambda)$ is invertible if it is $T$-dense on $[-K/2, K/2]$, and we can recover the signal using iterative algorithms. We are now interested in a more convenient way to recover our signal. If $\lambda$ satisfies some additional properties, then this reconstruction may be done in a more efficient way.


\subsection{$\lambda$ with finite support}

This method was introduced by Gröchenig and Schwab in \cite{Schwab03}. We suppose that the restriction of $\lambda$ over $[-K/2, K/2]$ has support of size $S$, and that $S << K$. Without loss of generality, we can suppose $Supp(\lambda) \subset [-S/2, S/2]$. We can think of polynomial splines for instance, since they have a small support. Let $J$ be the number of samples taken by our C-TEM inside $[-K/2,K/2]$, and let $M$ be the matrix of size $J \times K$ with elements $M_{jk} = \lambda(t_j -k)$. Let $\textbf{f}$ be the column vector of size $J$ with elements $f(t_j)$, and $\textbf{c} = (c_1, \cdots c_n)^T$ the vector representing our input signal: $f(t) = \sum_{k=1}^K c_k \lambda(t-k)$. Then:
\[
	\textbf{f} = M \textbf{c}
\]
and we can recover $\textbf{c}$ if and only if $M$ is left invertible. According to Lemma \ref{th:reconstruction2}, this is the case if $\{t_n\}$ is dense enough, and a left inverse may be calculated by:
\[
	M^\ddagger = (\mu M^T D M)^{-1}\mu M^T D .
\]
The positive definite matrix $U := M^T DM$ of size $K \times K$ has elements:
\begin{equation} \label{eq:U}
	U_{kl} = \sum_{j=1}^J w_j \lambda(t_j-k) \lambda(t_j-l) .
\end{equation}

\begin{lemma}
	A necessary condition to have $\lambda(t_j-k) \lambda(t_j-l) \neq 0$ is to have $| t_j -k | \mod K \le S/2$ and $| t_j -l | \mod K \le S/2$. In particular, if $ | k-l | \mod K > S$, then $U_{kl} = 0$.
\end{lemma}

\begin{proof}
	This follows from the fact that the support of $\lambda$ is included in $([-S/2, S/2] + K \mathbb{Z})$, and from the triangular inequality $ | k-l | \le |k-t_j| + |t_j - l|$.
\end{proof}

Therefore, most of the coefficients of the matrix $U$ are zero: $U$ is zero except on a large diagonal of size $S$, plus on its corners. Notice that if we work with finite support signals and not with periodic ones, then this property holds, without the modulo K. In this case, $U$ has just a large non null diagonal. In both cases, the matrix $U$ is invertible with a fast algorithm in $O(S^2)$, using the Choleski decomposition for sparse matrices. \\

\section{Robustness to noise}

With physical devices, time locations of the samples can only be recorded with finite precision, so that the effective recorded time $t_n'$ is different from the real time $t_n$. This may result from quantization error for instance. The purpose of this section is to show that time encoding machines are robust under this type of error. Note that we do not suppose the signal to be noisy, but only the locations of the samples.


\subsection{Some generalities about bounded errors}
We will study the error in the case of periodic signals. We suppose therefore that $\lambda$ is $K$-periodic. We recall the study so far in this finite case:

\begin{lemma}
	There exists $\tau$ such that, for all $T < \tau$, for all $\mathbf{t} = (t_1, \cdots, t_J)$ which is $T$-dense, the matrix $M(\mathbf{t})$ with elements $M(\mathbf{t})_{jk} = \lambda(t_j -k)$ of size $J \times K$ is left invertible.
\end{lemma}
We will denote in the following
\[
	N(\mathbf{t}) = (M(\mathbf{t})^T M(\mathbf{t}))^{-1} M(\mathbf{t})^T
\] 
to be the pseudo-inverse of $M(\mathbf{t})$. In this case, the reconstruction is given by:
\[
	\textbf{c} = N(\mathbf{t}) \Phi(\mathbf{t})
\]
where $\Phi(\mathbf{t}) = (\Phi_1(t_1), \cdots \Phi_J(t_J)) = (f(t_1), \cdots , f(t_J))$ for a Crossing-TEM. Because of the noise, we would have access to $\mathbf{t'} = (t_1', \cdots t_J')$, and thus to some $M(\mathbf{t'})$ and some $\Phi'(\mathbf{t'})$. The reconstruction from the noisy data is:
\[
	\textbf{c'} = N(\mathbf{t'}) \mathbf{\Phi'}(\mathbf{t'}) .
\]

Several problems may occur. The fact that the functions $\Phi_n$ may depend on the past is not well suited for our study, for we do not have a model to describe the effect of the noise on those functions. Therefore, we will study the case $\Phi_n = \Phi$, or $\mathbf{\Phi'}(\mathbf{t'}) = \mathbf{\Phi}(\mathbf{t'})$. Then, if we want $M(\mathbf{t'})$ to be left invertible, we should ensure that $\mathbf{t'}$ is $T'$-dense with $T' < \tau$. In particular, we cannot use an unbounded model for the noise in the timing domain. We introduce the infinite norm of $\mathbf{t} -\mathbf{t'}$ to be:
\[
	\| \mathbf{t} -\mathbf{t'} \|_\infty := \sup_{n=1..J} ( | t_n - t_n' |) ,
\]
and for $\epsilon > 0$, we denote the closed ball of center $\mathbf{t}$ and radius $\epsilon$ by
\[
	B(\mathbf{t}, \epsilon) := \left\{ \mathbf{t'} = (t_1', \cdots , t_J'), \ \| \mathbf{t} - \mathbf{t'} \|_\infty \le \epsilon \right\} .
\]

\begin{lemma} \label{th:T'}
	Let $T < T' < \tau$. If $\mathbf{t} = (t_1, \cdots t_J)$ is $T$-dense, with $T < \tau$, then, for all $\mathbf{t'} \in B(\mathbf{t}, \epsilon_0)$ where $\epsilon_0 = (T'-T)/2$, $\mathbf{t'}$ is $T'$-dense. In particular, $M(\mathbf{t'})$ is invertible for all $\mathbf{t'} \in B(\mathbf{t}, \epsilon_0)$.
\end{lemma}

\begin{proof}
	We use the triangular inequality:
	\[
		t'_{n+1} - t_n \le | t'_{n+1} - t_{n+1} | +  |t_{n+1} - t_n| + |t_n - t'_n| \le \epsilon_0 + T + \epsilon_0 \le T' .
	\]
\end{proof}		
	
If $J$ is too big, then the size of the matrices may explode. However, we may always suppose this number to be bounded by a constant:

\begin{lemma} \label{th:J0}
	Let $\mathbf{t} = (t_1, \cdots t_J)$ be $T$-dense. Then there exists $J' < 2 \lceil K/T \rceil$ and a sequence $(i_1, \cdots i_{J'})$ such that $\mathbf{\tilde t} = (t_{i_1}, \cdots t_{i_J'})$ is still $T$-dense. In particular, $M(\mathbf{\tilde t})$ is a submatrix of $M(\mathbf{t})$ of size $J' \times K$ which is left invertible.
\end{lemma}	

\begin{proof}
	With the notation $t_{J+1} = t_1 + K$, we construct $i_n$ by induction:
\[
	\begin{array}{lll}
		i_{1} & = & 1 \\
		i_{n+1} & = & \max \{j \in [ 1 \cdots J+1], \ |t_j - t_{i_1}| \le T \}
	\end{array}.
\]
This sequence is constant from a certain rank $J'+1$, and we have $t_{i_{J'+1}} = t_{J+1} = t_{i_1} + K$. Moreover, with this definition, we always have $t_{i_{n+2}}- t_{i_{n}} \ge T$. Finally, we write:
	\[
		K = t_{i_{J'+1}} - t_{i_1} = (t_{i_{J'+1}} - t_{i_{(2 \lfloor J'/2 \rfloor +1)}}) +  ( t_{i_{(2 \lfloor J'/2 \rfloor +1)}} - t_{i_{(2 \lfloor J'/2 \rfloor -1)}}) + \cdots + (t_{i_3} - t_{i_1}) > \lfloor J'/2 \rfloor \cdot T .
	\]
Therefore, $(K/T) > \lfloor J'/2 \rfloor$, $\lceil K/T \rceil > J'/2$ and $J' < 2 \lceil K/T \rceil$.
\end{proof}

Because $M(\mathbf{\tilde t})$ is left invertible, of size $J' \times K$, we can use it to reconstruct the signal: $ \textbf{c'} = N(\mathbf{\tilde t}) \Phi(\mathbf{\tilde t})$. Thus, we can always suppose the number of lines of $M$ to be less than $2 \lceil K/2 \rceil$ without loss of generality. Note that because $M$ is full column rank, there exists $K$ lines such that the $K\times K$ resulting sub-matrix is invertible. We could have concluded that this sub-matrix is still invertible whenever $\mathbf{t}$ is taken in a small neighborhood, for the set of invertible matrices is an open set. However, we do not have any expression of this neighborhood. This is why we have chosen to take more lines, to guarantee $T$-density, and therefore left-invertibility on a large neighborhood, according to Lemma \ref{th:T'}.


\subsection{Robustness under bounded noise in the timing domain}

In sampling theory for band-limited signals, the error introduced by quantization in the amplitude domain leads to an error $MSE(X', X) = O(\epsilon^2)$, where $X'$ is the reconstruction from the noisy data, $X$ is the original data, and $MSE$ stands for the Mean Square Error: $MSE(X',X) = \|X' - X\|_{l^2}^2$. Thao and Vetterli have shown in \cite{Thao94} that we also have $MSE(X',X) = O(r^{-2})$ where $r$ denotes the oversampling rate. This theorem has been extended by Chen, Han and Jia in \cite{Chen05b, Chen05} for non-uniform sampling in shift-invariant subspaces. Using the mean value theorem, these articles use the fact that having oversampling is equivalent to having infinite precision of the amplitude of some $f(t_n)$, where the locations $\{t_n\}$ are not known precisely (but are more and more precise with the oversampling factor). Here, we are working in the framework of TEMs in a finite dimensional space.  In this case, we also have $MSE(X', X) = O(\epsilon^2)$. The calculations are similar to the ones in \cite{Thao94}.

\begin{theorem} \label{th:noise}
	Let $\lambda$ be of class $C^1$ and $K$-periodic. Let $\mathcal{CT}_{\Phi}$ be a C-TEM on $V^2_K(\lambda)$ without feedback which is $T$-dense, with $T < \tau = \pi \inf (G_\lambda / G_{\lambda'})$. Let $T'$ such that $T < T' < \tau$ and $\epsilon_0 = (T' - T)/2$. Then, there exists a constant $\alpha$ which depends only on $\lambda, \Phi, K$ and $T'$ such that, for all output $\mathbf{t} = (t_1, \cdots t_J)$ of $\mathcal{CT}_\Phi$, for all $\mathbf{t'} = (t_1', \cdots t_J') \in B(\mathbf{t}, \epsilon_0)$, we have:
\[
	\| N(\mathbf{t'}) \Phi(\mathbf{t'}) - N(\mathbf{t}) \Phi(\mathbf{t}) \|_{l^2} \le \alpha \|\mathbf{t'} - \mathbf{t}\|_\infty .
\]
\end{theorem}

\begin{proof}
	According to Lemma \ref{th:J0}, we can always suppose $J = J_0 = 2 \lceil K/T \rceil$. We introduce $\epsilon = \|\mathbf{t'} - \mathbf{t}\|_\infty < \epsilon_0$. Let $S_\delta := \{ \mathbf{t} = (t_1, \cdots t_{J_0}), t \text{ is increasing}, \mathbf{t} \text{ is $\delta$-dense} \}$. According to Lemma \ref{th:T'}, 
\[
	S_T + B(0, \epsilon_0) \subset S_{T'}	
\]
and $S_{T'}$ is compact (it is closed and bounded in $\mathbb{R}^{J_0}$). Lemma \ref{th:T'} ensures that $N(\mathbf{t})$ is well-defined in $S_{T}+B(0, \epsilon_0)$, and its expression is:
\[
	N(\mathbf{t}) = (M(\mathbf{t})^T M(\mathbf{t}))^{-1} M(\mathbf{t})^T .
\]
In particular, the coefficients of $N(\mathbf{t})$ are continuous with respect to the coefficients of $M(\mathbf{t})$. Therefore, if $\lambda$ is of class $C^1$, so are the coefficients of $N(\mathbf{t})$. Actually, if we denote by $N_{kj}(\mathbf{t})$ the coefficients of $N(\mathbf{t})$, and $F^i_{kj}(\mathbf{t}) = \frac{\partial}{\partial t_i} N_{kj}(\mathbf{t})$ the coefficients of $F^i(\mathbf{t}) = \frac{\partial}{\partial t_i} N(\mathbf{\mathbf{t}})$, then all those functions are continuous on the compact $S_{T'}$. Thus, we can find a constant $\alpha_1$ such that:
\[
	\forall \mathbf{t} \in S_{T'}, \quad | N_{kj}(\mathbf{t}) | \le \alpha_1 \quad \text{and} \quad | F^i_{kj}(\mathbf{t}) | \le \alpha_1 .
\]
For a matrix $A = ( A_{ij} )$ of size $K \times J_0$, and for $b \in \mathbb{R}^{J_0}$, using the Cauchy-Schwartz inequality, we have:
\begin{equation} \label{eq:ineq}
	\| Ab \|^2_{l^2} = \sum_{i=1}^K \Bigl( \sum_{j=1}^{J_0} A_{ij} b_j \Bigr)^2 \le   \sum_{i=1}^K \Bigl( \sum_{j=1}^{J_0} A_{ij}^2\Bigr) \cdot   \Bigl( \sum_{j=1}^{J_0} b_j^2\Bigr) \le K\cdot J_0^2 \cdot (\max A_{ij}^2) \cdot (\max b_j^2) .
\end{equation}
We also have the Taylor inequality:
\begin{equation} \label{eq:Taylor1}
		 | \Phi(t_i') - \Phi(t_i) | \le | t_i' - t_i | \cdot \sup_{\tilde t \in (t_i, t_i')} \Phi'(\tilde t) \le \epsilon \| \Phi' \|_\infty
\end{equation}
and
\begin{equation} \label{eq:Taylor2}
		 | N_{ik}(t') - N_{ij}(t) |  \le \sum_{j=1}^{J_0} |t_j' - t_j| \cdot \sup_{\tilde t \in (t_j, t_j')} | F^j_{ik}(\tilde t) |  \le \alpha_1 J_0 \epsilon .
\end{equation}

We now write 
\[
	\| N(\mathbf{t'}) \Phi(\mathbf{t'}) - N(\mathbf{t}) \Phi(\mathbf{t}) \|_{l^2} \le \| N(\mathbf{t'}) \bigl( \Phi(\mathbf{t'}) - \Phi(\mathbf{t}) \bigr) \|_{l^2} + \| \bigl( N(\mathbf{t'}) - N(\mathbf{t}) \bigr) \Phi(\mathbf{t})\|_{l^2} .
\]
Using inequalities (\ref{eq:ineq}) and (\ref{eq:Taylor1}), the first term is bounded above by:
\[
	 \| N(\mathbf{t'}) \bigl( \Phi(\mathbf{t'}) - \Phi(\mathbf{t}) \bigr) \|^2_{l^2} \le K  J_0^2  \cdot \max N_{ij}(\mathbf{t})^2 \cdot \max{(\Phi(t_i') - \Phi(t_i))^2} \le (K J_0^2 \alpha_1^2 \| \Phi '\|^2_\infty) \epsilon^2 .
\]
The second term is bounded in the same way with (\ref{eq:ineq}) and (\ref{eq:Taylor2}):
\[
	\| \bigl( N(\mathbf{t'}) - N(\mathbf{t}) \bigr) \Phi(\mathbf{t})\|^2_{l^2} \le K J_0^2 \cdot \max(N_{ij}(\mathbf{t'}) - N_{ij}(\mathbf{t}))^2 \cdot \max \Phi(t_i)^2 \le (K J_0^4 \alpha_1^2 \| \Phi \|_\infty^2) \epsilon^2 .
\]

Altogether, we proved that
\[
	\| N(\mathbf{t'}) \Phi(\mathbf{t'}) - N(\mathbf{t}) \Phi(\mathbf{t}) \|_{l^2} \le \alpha \cdot \epsilon \quad \text{with} \quad \alpha := \sqrt K J_0 \alpha_1 \sqrt{ \| \Phi' \|_\infty^2 + J_0^2 \|  \Phi \|_\infty^2} .
\]

\end{proof}


\subsection{Numerical Results}

We have implemented the method described above with the following parameters. We took $f(t) = \sum_{k=1}^K c_k \lambda(t-k)$ with $K=50$, $\lambda$ the cardinal spline of order $3$, and $c_k$ are chosen uniformly at random in $[-1, 1]$. Then, we normalized $f(t)$ so that $\| f \|_{\infty} = 1$. We encode such a signal with the Crossing-TEM $\mathcal{CT}_\Phi$ with:
\[
	\Phi(t) = \alpha \cdot \cos \Bigl(  2 \pi t   \Bigr),
\]
with $\alpha = 1.1$. It is easy to check that this TEM is $1$-dense, and therefore is invertible according to \cite{Aldroubi99}. The crossing times $(\tilde t_1, \cdots, \tilde t_J)$ are then recorded with a quantization error $\epsilon$. From those points, we calculate the matrix $M$ with elements $M_{ik} = \lambda(\tilde t_i - k)$, and the vector $\mathbf{\tilde f} = (\Phi(\tilde t_1), \cdots, \Phi(\tilde t_n))$, and we recover $\mathbf{\tilde c} = (\tilde c_1, \cdots \tilde c_K)$ with $\mathbf{\tilde c} = M^{\dagger} \mathbf{\tilde f}$. We then have the reconstructed signal $\tilde f(t) = \sum_{k=1}^K \tilde c_k \lambda(t-k)$. We then recorded the $L^2$ error between $f$ and $\tilde f$ for various $f$. In Figure 2, we plotted the mean and the 95\% confidence interval of the $L^2$ error for $1000$ different signals $f(t)$, with respect to the $l^2$ norm of $\{\tilde t_n - t_n\}_{n = 1 .. J_0}$. We recall that because the number of samples is finite, this $l^2$ error is equivalent to the $l^{\infty}$ error. The choice  of the $l^2$ quantization error leads to a linear behavior. Therefore, the $L^2$ error of the signals with respect to the $l^{\infty}$ quantization error of the time locations is sub-linear, as Theorem \ref{th:noise} stated.

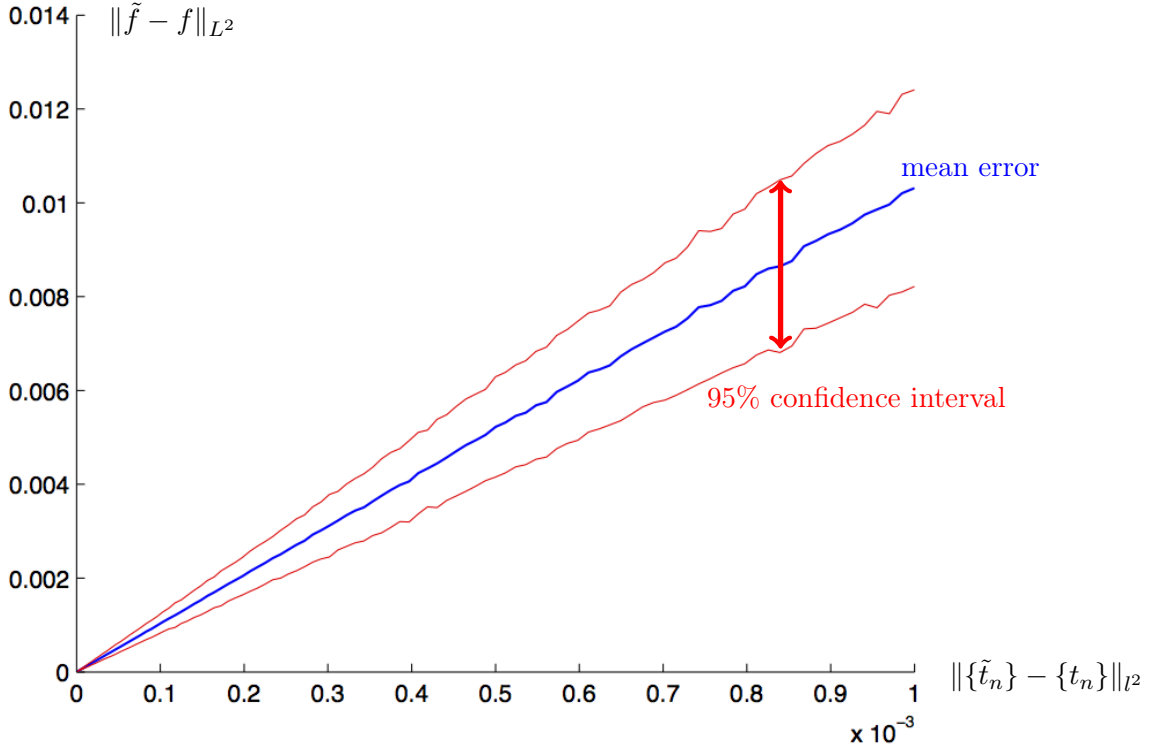
\begin{figure}[!h]
\label{fig:noise}
\begin{center}

\begin{tikzpicture}
	\node at (0,0) {\includegraphics[scale=0.9]{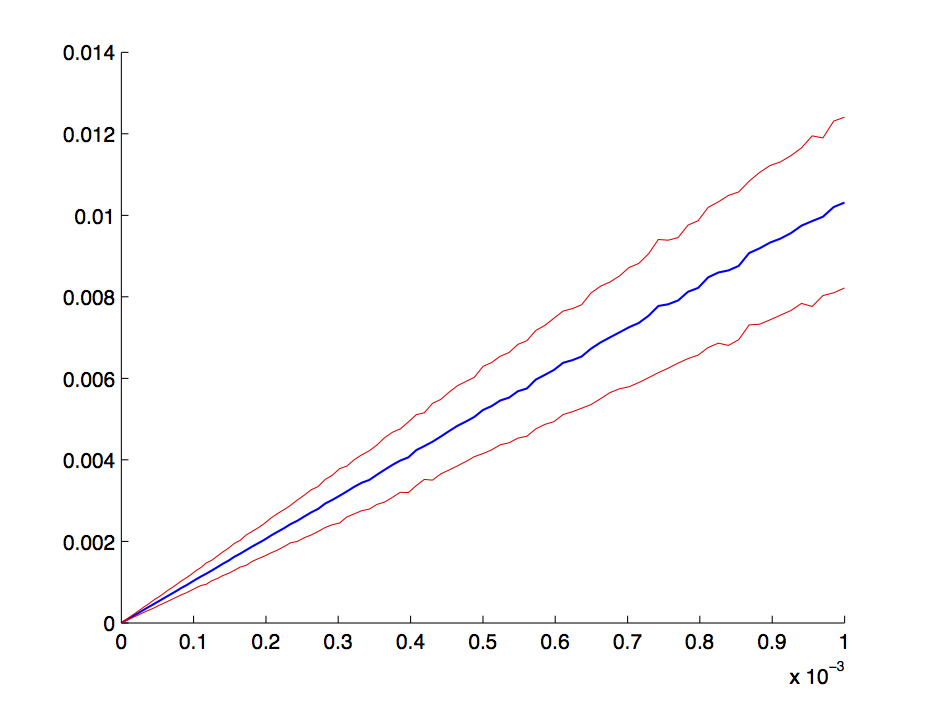}};
	\node at (6.5, 2.5) {\textcolor{blue}{mean error}};
	
	\node at (4, 0) (a) {};
	\node at (4,2.5) (b) {};
	\draw[<->] (a) to (b) [color=red, line width=2];
	
	\node at (5, -0.5) {\textcolor{red}{95\% confidence interval}};
	
	\node at (7.5, -4.2) {$\| \{\tilde t_n\} - \{t_n\} \|_{l^2}$};
	\node at (-4, 4.5) {$ \| \tilde f - f \|_{L^2}$};
\end{tikzpicture}

\end{center}
\caption{The $L^2$ error of the signal versus the $l^2$ quantization error of the time locations.}
\end{figure}


\section{Future Work}
The tools we used to prove the invertibility of dense C-TEMs are the following facts: $V^2(\lambda)$ is a RKHS, there exists an orthogonal projection onto this space, we can use Wirtinger's inequality on this space, and finally, there is an inequality of the form:
\[
	\forall f \in V^2(\lambda), \| f'\|_{L^2} \le C \| f \|_{L^2}
\]
saying that the map:
\[
	\begin{array}{lcll}
		D : & V^2(\lambda) \subset H^1 & \to & L^2 \\
		 & f & \mapsto & f'
	\end{array}
\]
is continuous. Therefore, similar results may be derived on a space $F$ where those requirements are satisfied. We have the general
\begin{theorem} \label{th:general}
	If $F$ is a closed subspace of $H^1$ such that the norms $\| . \|_{H^1}$ and $\| . \|_{L^2}$ are equivalent on $F$, i.e. there exists a constant $C$ such that:
	\[
		\forall f \in F, \ \| f'\|_{L^2} \le C \cdot \| f \|_{L^2} ,
	\]
	then, every C-TEM (and IF-TEM) which is $T$-dense, with $T < \tau := \pi / C$, is invertible.
\end{theorem}

\begin{proof}
	Because $F \subset H^1 \subset L^2$ is closed, we already have the existence of an orthogonal projector on $F$. We can use the Wirtinger's inequality for $F \in L^2$. Finally, $F$ is a RKHS because $H^1$ is a RKHS for its norm (with kernel $K_t(u) = e^{- | t-u |}/2$), and because the $H^1$ norm and $L^2$ norm are equivalent on $F$.	Finally, the exact same proof of Theorem \ref{th:main} and Theorem \ref{th:mainbis} can be derived, which gives the result.
\end{proof}
Therefore, our result about TEMs on SISS can probably be extended in this direction. However, while it is easy to show that $V^2(\lambda)$ satisfies the conditions of Theorem \ref{th:general}, it may be more difficult to extend it to some other cases. Moreover, the case of SISS allowed us to deduce the algorithms for this case. 

Another direction to look at for future work is the behavior of the error. While we showed that the error is sub-linear in the $l^{\infty}$ quantization error, we noticed that this error appears to be actually linear in the $l^2$ quantization error. We could not prove this fact, nor could we give an explicit bound of the decay of the error. 

A final direction, and probably the most difficult one is the following remark. If $\lambda$ has finite support $[0, S]$, then all the elements of the matrix $M$ are zero, except on a large diagonal, and the matrix $M$ is left-invertible. Suppose that the TEM has sampled $J$ times during the interval $[k_1, k_2]$, then the value of the function on this interval depends on $K = k_2 - k_1 + S$ coefficients. In particular, if the TEM has density $T < 1$, and if the interval is large enough, we can ensure that $J \ge K$, so that there are more equations than coefficients involved. But we cannot conclude that the sub-matrix involving those equations and coefficients is still left-invertible. Aldroubi and Gröchenig showed that it is the case for cardinal splines SISS, because of the very special structure of those functions \cite{Aldroubi99}. This would allow us to reconstruct the signal independently of the past, and therefore provides us better algorithms for reconstruction, or bounds for the error.




\appendix

\bibliographystyle{plain}
\bibliography{SamplingBasedOnTiming}







\end{document}